\documentclass[a4paper,10pt]{scrartcl}
\usepackage{amssymb}
\usepackage{array}
\usepackage{graphicx}
\usepackage[english]{babel}
\usepackage[utf8]{inputenc}
\usepackage{amsmath}
\usepackage{amsthm}
\usepackage{units}
\usepackage{dsfont}
\usepackage[arrow, matrix, curve]{xy}
\usepackage[left,modulo]{lineno}
\newcommand{\tu}[1]{\textup{#1}}

\newcommand{\Abb}[4]{\left\{ \begin{array}{ccc}
                               #1 & \rightarrow &#2\\
			       #3 &\mapsto &#4
                               \end{array}\right.}
\newcommand{\supp}{\textup{supp}}
\newcommand*{\operp}{\perp\mkern-20.7mu\bigcirc}

\newcommand{\N}{\mathbb{N}}
\newcommand{\R}{\mathbb R}
\newcommand{\Z}{\mathbb Z}
\newcommand{\C}{\mathbb{C}}

\newcommand{\g}{\mathfrak g}
\newcommand{\f}{\mathfrak f}

\theoremstyle{plain}
\newtheorem{Satz}{Satz}
\newtheorem{lem}[Satz]{Lemma}
\newtheorem{prop}[Satz]{Proposition}
\newtheorem{thm}[Satz]{Theorem}
\newtheorem{cor}[Satz]{Corollary}
\theoremstyle{definition}

\newtheorem{Def}{Definition}
\newtheorem{hyp}{Hypothesis}
\newtheorem{exmpl}{Example}

\newtheorem{innercustomhyp}{Hypothesis}
\newenvironment{customhyp}[1]
  {\renewcommand\theinnercustomhyp{#1}\innercustomhyp}
  {\endinnercustomhyp}

\theoremstyle{remark}
\newtheorem{rem}{Remark}

\numberwithin{equation}{section}
\numberwithin{Satz}{section}
\numberwithin{exmpl}{section}
\numberwithin{hyp}{section}
\numberwithin{rem}{section}

\title{Equivariant spectral asymptotics for $ h$-pseudodifferential operators}
\author{Tobias Weich}

\begin{document}

\maketitle

\begin{abstract}
We prove equivariant spectral asymptotics for $ h$-pseudodifferential operators for compact orthogonal group actions generalizing results of El-Houakmi and Helffer \cite{EH91} and Cassanas \cite{Cas06}. Using recent results for certain oscillatory integrals with singular critical sets \cite{Ram14} we can deduce a weak equivariant Weyl law. Furthermore, we can prove a complete asymptotic expansion for the Gutzwiller trace formula without any additional condition on the group action by a suitable generalization of the dynamical assumptions on the Hamilton flow.
\end{abstract}
\tableofcontents
\section{Introduction}

Semiclassical analysis  provides a rigorous way to establish a connection between quantum physics and classical mechanics \cite{Rob87,Zw12}. Mathematically, a non-relativistic quantum mechanical system is described by the spectrum of an $h$-pseudodifferential operator ($h$-PDO) $\hat H$ on $\R^n$ while the classical system is described by the Hamilton-flow of the associated symbol $H(x,\xi)$ on the symplectic vector space $T^*\R^n$. From a mathematical point of view the question of connecting quantum physics to classical mechanics thus consists in proving asymptotic formulae for the spectrum of $\hat H$ in terms of $H$ in the limit $ h \to 0$.

The two main theorems which have been established in this field are the Weyl law and the Gutzwiller trace formula. While the Weyl law gives an estimate for the number of eigenvalues below a certain energy in terms of the phase space volume of the corresponding classical energy shell, the Gutzwiller trace formula gives an asymptotic formula for the correlation of the eigenvalues in  terms of periodic orbits of the classical flow. 

The first version of Weyl's law was already proven in 1912  for special second order partial differential equations by Hermann Weyl \cite{We12}, who wanted to give a rigorous foundation for the Rayleigh-Jeans law for black body radiation. This work has been generalized in various different works, among which we want to mention the work of Hörmander \cite{Hor68}, who extended these results to elliptic pseudodifferential operators on compact manifolds, using the approach of Fourier integral operators (FIOs), and the work of Helffer and Robert \cite{HR83b} in the context of elliptic $ h$-PDOs on $\R^n$. 

The connection between spectral correlations and periodic orbits was first established in physics by the pioneering work of Gutzwiller \cite{G71} and Balian and Bloch \cite{BB69} around 1970. Shortly later, between 1973 and 1975, similar results were rigorously proven for elliptic operators on compact manifolds by Colin de Verdière \cite{CdV73I,CdV73II,CdV07}, Chazarain \cite{Cha74}, and Duistermaat and Guillemin \cite{DG75}. For $ h$-PDOs, rigorous proofs have finally been given in the 90s by Paul and Uribe \cite{PU91,PU95} using a trick of Colin de Verdière \cite{CdV79} which allows to link the semiclassical problem to homogeneous PDOs, and by Meinrenken using a direct FIO approach \cite{Mei92}.

Symmetries of physical systems are known to simplify their study significantly. If for example the Hamiltonian $\hat H$ commutes with a compact Lie group $G$ acting linearly on $\R^n$, then $\hat H$ leaves the Peter-Weyl decomposition $L^2(\R^n)=\bigoplus_{\chi\in \hat G} L^2_\chi$ of this group action invariant. Accordingly, one can restrict the Hamilton operator to one of the subspaces $L^2_\chi$ belonging to an irreducible representation $\chi$ of $G$. One can then study the spectral asymptotics of the symmetry-reduced Hamilton operator $\hat H_\chi:=\hat H_{|L^2_\chi}$. This amounts to the question of counting the eigenvalues of this reduced Hamilton operator (equivariant Weyl law) and the question of correlations in the spectrum of $\hat H_\chi$ (equivariant Gutzwiller trace formula).

For elliptic operators on compact manifolds, the leading term of the equivariant Weyl law has been obtained for general group actions by Donnelly \cite{Don78} and Brüning-Heintze \cite{BH79} using heat kernel methods, which however do not allow to obtain remainder estimates. First remainder estimates were only obtained for actions with only one orbit type by Donnelly \cite{Don78}, Brüning-Heintze \cite{BH79}, Brüning \cite{Bru83} and Guillemin-Uribe \cite{GU90} in the homogeneous setting, and by Helffer-Robert \cite{HR84,HR86} and El-Houakmi-Helffer \cite{EH91} for $ h$-PDOs. When trying to obtain remainder estimates for general group actions, it turns out that singularities in the critical set of the appearing phase function cause serious problems. Recently it has been possible to obtain remainder estimates for the equivariant Weyl law for elliptic operators on compact manifolds which are invariant under general compact group actions \cite{Ram14} by partially resolving the singularities of the critical set (
see also \cite{CR09} for preliminary results). 

In the derivation of equivariant Gutzwiller formulae, singularities in the critical set are a major obstruction in deriving the desired spectral asymptotics, too. There are two sources for these singularities: The group action and the Hamilton dynamics. Because of this, in all articles on equivariant Gutzwiller trace formulae (Guillemin-Uribe \cite{GU90} for elliptic PDOs on compact manifolds and Cassanas \cite{Cas06fg,Cas06} for $ h$-PDOs) the spectral asymptotics only hold under two conditions: The clean-flow or non-degenerate orbit condition, which are conditions on properties of the Hamilton flow (see e.g. \cite[Definition 4.4]{Cas06}), and the condition of reduction, which is a condition on the group action (see e.g. \cite[Definition 2.1]{Cas06}).

In the present work, equivariant spectral asymptotics of $ h$-PDOs for general compact group actions are derived. Concerning the number of eigenvalues, we prove a  weak equivariant Weyl law by reducing the problem to oscillatory integrals with singular critical sets which already appeared in \cite{Ram14}. Regarding the spectral correlations, we drop the condition on the group action on which the clean flow and non-degenerate orbit conditions are usually based. Instead we generalize the non-degenerate orbit condition for arbitrary group actions. Based on this new generalized condition on the Hamilton flow, we can derive a complete asymptotic expansion for the equivariant Gutzwiller trace formula without additional reduction assumptions. Furthermore we present a simple explicit example of a 3-dimensional Harmonic oscillator which illustrates that these new conditions are much less restrictive. 

\textbf{Acknowledgements:} I like to thank Pablo Ramacher for introducing me to this topic and for the support and feedback on this work. I also acknowledge helpful discussions with Yves Colin de Verdière during my stay at Grenoble which was afforded by the ANR via the project 2009-12 METHCHAOS. This work has been made possible by a grant of the German National Academic Foundation.

\section{Definitions and Notations}
\subsection{Classical Hamilton mechanics and quantization}
We will consider Hamilton systems on the phase space $\R^{2n}$. Points in phase space will commonly be denoted by $z\in \R^{2n}$, respectively by $(x,\xi)\in \R^{2n},~x,\xi\in \R^n$, if we want to distinguish between the space variable $x$ and momentum variable $\xi$. Even if the tangent space $T_z\R^{2n}$ is canonically isomorphic to $\R^{2n}$ we will emphasize the difference by denoting tangent vectors by Greek letters $\alpha, \beta\in T_z\R^{2n}$. The $2n\times2n$ matrix $J=\left(\begin{array}{cc}0&\mathds 1\\-\mathds 1&0\end{array}\right)$ defines  a symplectic form by $\omega(\alpha,\beta) := \langle\alpha,J\beta\rangle$ where $\langle,\rangle$ is the standard scalar product on $\R^{2n}$. 

A Hamiltonian on the symplectic manifold $(\R^{2n},\omega)$ is a smooth function $H:\R^{2n}\to \R$. Via the symplectic structure we can associate a Hamilton vector field $X_H$ to the Hamiltonian $H$ by the condition $dH(\bullet)=\omega(X_H,\bullet)$. For the trivial symplectic manifold ($\R^{2n},\omega$) the Hamilton vector field is simply given by $J\nabla H(z)$, where $\nabla$ denotes the gradient in $\R^{2n}$. This vector field generates the Hamilton flow  $\Phi_t: \R^{2n}\to \R^{2n}$ by setting 
\begin{equation}\label{eq:HamiltonEquations}
\tfrac{d}{dt}\Phi_t(z)=J\nabla H(\Phi_t(z)).
\end{equation} 
In order to simplify notation we will sometimes write $(x_t,\xi_t)=z_t=\Phi_{t}(z)$. 

For a given energy $E\in \R$, the energy shell is denoted by $\Sigma_E:=H^{-1}(E)$. As $\omega$ is antisymmetric, we have $0=\langle \nabla H(z_t),J\nabla H(z_t)\rangle =\tfrac{d}{dt} H(z_t)$. Thus, the energy is preserved along the flow, and the Hamilton flow leaves $\Sigma_E$ invariant. Furthermore, $\nabla H(z)\perp J\nabla H(z)$, and if $H$ is non-degenerate in $z$, i.e.~if $\nabla H(z)\neq 0$, then $\Sigma_E$ is locally smooth and $T_z\Sigma_E=(\nabla H(z))^\perp$. 

Using a general formula for the Lie-derivative of a differential form we calculate
\[
 \frac{d}{dt}_{|t=0} (\Phi_t^*\omega)=\mathcal L_{X_H} \omega = d(X_H\lrcorner \omega) + X_H\lrcorner d\omega = ddH = 0
\]
where we used the closeness of $\omega$ and the definition of $X_H$. Consequently the symplectic form is preserved under the Hamilton flow and the flow consists thus of symplectomorphisms. 

In order to discuss the quantization of the classical Hamiltonian, we first introduce the notion of order functions and symbol classes (see e.g. \cite[Chapter 4]{Zw12} for a more detailed introduction).
\begin{Def}[Order function]
A measurable function $m:\R^{2n}\to\R^+$ is called \emph{order function} if there exist constants $C>0$ and $N\in \R$ such that
\[
 m(z)\leq C\langle z-z'\rangle^N m(z')
\]
where $\langle z\rangle:=\sqrt{1+|z|^2}$.
\end{Def}
A standard example for such an order function is e.g. $m(x,\xi)=\langle x\rangle^a\langle \xi\rangle^b$ with $a,b\in \R$.
\begin{Def}[Symbol class]
 Let $m$ be an order function. Then the corresponding \emph{symbol class} is defined by
 \[
  S(m):=\{a\in C^\infty(\R^{2n}),~\forall\,\gamma\in \N^{2n} \,\exists\,C_\gamma \tu{ with }|\partial^\gamma a|\leq C_\gamma m\}.
 \]
\end{Def}
\begin{Def}[Quantization]
For a symbol $H\in S(m)$ and $t\in[0,1]$, the \emph{quantization} of $H$ is defined on the Schwartz space $\mathcal S(\R^{n})$ of rapidly decreasing functions as the operator
\[
 \Big(\tu{Op}_{ h,t}(H) \Psi\Big)(x):= (2\pi h)^{-n}\int\limits_{\R^{2n}} H\left(tx+(1-t)y,\xi\right)\Psi(y)e^{\frac{i}{ h}(x-y)\xi} dyd\xi.
\]
Two particular important special cases are the standard quantization for $t=1$ 
\[
 \Big(\tu{Op}_{ h}(H) \Psi\Big)(x):= (2\pi h)^{-n}\int\limits_{\R^{2n}} H\left(x,\xi\right)\Psi(y)e^{\frac{i}{ h}(x-y)\xi} dyd\xi
\]
and the Weyl quantization for $t=\tfrac{1}{2}$
\[
 \Big(\tu{Op}_{ h}^w(H) \Psi\Big)(x):= (2\pi h)^{-n}\int\limits_{\R^{2n}} H\left(\frac{x+y}{2},\xi\right)\Psi(y)e^{\frac{i}{ h}(x-y)\xi} dyd\xi.
\]
\end{Def}
The Weyl quantization has the particular advantage that real symbols give rise to $L^2$-symmetric operators. We will therefore mainly use the Weyl quantization, and in order to shorten the notation, we will denote the Weyl quantization of a symbol by $\hat H:=\tu{Op}_ h^w(H)$. In general, such a quantized pseudodifferential operator will, however, not have purely discrete spectrum. In order to obtain a Hamilton operator with satisfactory spectral properties we suppose
\begin{hyp}\label{hyp:regularityH}
 $H:\R^{2n}\to \R$ is a smooth function bounded from below such that $\langle H\rangle$ is an order function and $H\in S(\langle H\rangle)$.
\end{hyp}
From these regularity conditions it follows with \cite[Théorème 2.6]{HR83} that $\hat H$ extends to a unique self-adjoint operator on the space of square integrable functions $L^2(\R^n)$. However, such an operator will not have discrete spectrum in an interval $[E_1,E_2]$ in general, and therefore we assume
\begin {hyp}\label{hyp:closed}
 There are constants $E_1<E_2$ and $\varepsilon >0$ such that $H^{-1}([E_1-\varepsilon, E_2+\varepsilon])\subset \R^{2n}$ is compact.
\end {hyp}
The physical interpretation of this condition is that the system is closed, and from \cite[Proposition 5.1]{HR83} we obtain that there is an $ h_0>0$ such that for $ h\in ]0, h_0]$ the spectrum of $\hat H$ in the interval $[E_1,E_2]$ is purely discrete. 

\subsection{Group actions and symmetry reduction}\label{sec:group_action}
We first recall some general facts about compact group actions. Let $G$ be a compact Lie group acting on a smooth manifold $M$. For $g\in G$ and $m\in M$ we will simply denote the group action by $gm$. Given a point $m \in M$, we define its stabilizer subgroup as $G_m:=\{g\in G,~gm=m\}\subset G$. The orbit of $m$ under the group action will be denoted by $Gm:=\{gm, g\in G\}\subset M$, and it is always a smooth manifold diffeomorphic to the homogeneous space $G/G_z$. Furthermore, all the stabilizer subgroups $G_{m'}$ for arbitrary points on the orbit $m'\in Gm$ are conjugate to each other. In general, for points on different orbits this is not the case, which leads to 
\begin{Def}[Orbit types]
Let $m, m'\in M$ be two points. Their orbits $G/G_m$ and $G/G_{m'}$ are said to be \emph{of the same type} if the stabilizer subgroups $G_m$ and $G_{m'}$ are conjugate in $G$. If two orbits $G/H_1$ and $G/H_2$ ($H_1$ and $H_2$ being closed subgroups of $G$) are given, then one can define a \emph{partial order} on the orbit types by defining
\begin{equation}\label{eq:OrbitOrder}
 \textup{type}(G/H_1)\leq \textup{type}(G/H_2) \Leftrightarrow H_2 \textup{ is conjugate to a subgroup of } H_1.
\end{equation}
\end{Def}
One of the most important results in the theory of compact group actions is the existence of a principal orbit type.
\begin{prop}\label{prop:principalOrbit}
 Let $G$ be a compact group acting smoothly on a manifold $M$. Then there exists a principal orbit type $G/H$ such that the union of all points of this orbit type $M(H):=\{m\in M,~G_m \tu{ is conjugate to } H\}$ is open and dense in $M$. Furthermore, this orbit type is maximal with respect to the order relation (\ref{eq:OrbitOrder}).
\end{prop}
\begin{proof}
 See e.g. \cite[Theorem IV.3.1]{Bre72}.
\end{proof}

Symmetry groups of a Hamilton system are always required to preserve the symplectic structure. For a Hamilton system on $\R^{2n}$ a linear group $G$ preserving the symplectic structure $\omega$ (i.e.~fulfilling $\omega(g\alpha,g\beta)=\omega(\alpha,\beta)$ for all $\alpha,\beta\in \R^{2n}$ and $g\in G$) has thus to be a subgroup of the symplectic group $Sp(n)$.
An important class of symplectic group actions arise from subgroups $G\subset Gl(n,\R)$. For an element $g\in G$ acting linearly on $\R^n$, this action can be lifted to a symplectic action on $\R^{2n}$ by defining $g(x,\xi):=(gx,(g^t)^{-1}\xi)$. In the sequel, we will only consider compact subgroups $G\subset O(n)$, so that the symplectic action simplifies to $g(x,\xi)=(gx,g\xi)$. 
Studying orthogonal subgroups is in fact not more restrictive than studying compact subgroups $G\subset GL(n,\R)$, because by a standard averaging argument each such group $G$ is conjugate to an orthogonal subgroup $S_0 G S_0^{-1}\subset O(n)$ for a suitable $S_0\in Gl(n,\R)$ (see e.g. \cite[Lemma 4.6]{Cas06fg} for more details). Through the smooth path $e^{At}z\subset \R^{2n}$, each element $A$ of the Lie algebra $\g$ of $G$ defines a vector in $T_z\R^{2n}$. Identifying the tangent space with $\R^{2n}$, and 
considering $\g$ as a matrix Lie algebra of $n\times n$ matrices, this vector is exactly given by the matrix action $(Ax,A\xi)$. We will generally denote it by $Az$ and write the vector space spanned by this Lie algebra action as $\g z:=\{Az,~A\in \g\}$. It is straightforward to check that $\g z=T_z(Gz)$. 

A Hamiltonian is said to be $G$-invariant if $H(gz)=H(z)$ for all $g\in G, z\in \R^{2n}$. In this case the flow commutes with the $G$-action. Indeed, under certain assumptions such Hamilton systems can be reduced, i.e.~one can construct a lower dimensional Hamilton system by using constants of motion and dividing out the group action. An important object for this symplectic reduction is the so called momentum map, which encodes the constants of motion. In our setting it is given by
\[
 \mu:\Abb{\R^{2n}}{\g^*}{z}{\mu(z)} \tu{ where }\mu(z)(A):=\langle z,JAz\rangle.
\]
The zero level of the momentum map is given by 
\[
\Omega_0:=\mu^{-1}(0):=\{z\in \R^{2n}, \langle z, JAz\rangle =0~~\forall A\in \g\} .
\]
The regularity of $\Omega_0$ strongly depends on the properties of the $G$-action. If $U\subset \R^{2n}$ is an open $G$-invariant subset intersecting $\Omega_0$ such that all orbits in $\Omega_0\cap U$ are of the same type, then it is known (see e.g.~\cite[Proposition 2]{CR09}) that $\Omega_0\cap U$ is a smooth manifold with tangent space
\[
 T_z(\Omega_0\cap U) = (J\g z)^\perp.
\]
Furthermore, one has the following important
\begin{thm}[Symplectic reduction]\label{thm:SymplRed}
 If $U$ is a set with the properties as above, then $G\backslash(\Omega_0\cap U)$ carries a unique structure of a smooth symplectic manifold and the projection
\[
 Pr:\Abb{\Omega_0\cap U}{G\backslash(\Omega_0\cap U)}{z}{[z]}
\]
which associates to each point $z$ the corresponding equivalence class $[z]$ in the quotient is a smooth submersion. If $H:U\to\R$ is a $G$-invariant Hamiltonian, then it reduces to a smooth Hamiltonian $\tilde H:G\backslash(\Omega_0\cap U)\to \R$, and the Hamilton flow $\tilde \Phi_t$ of $\tilde H$ on $G\backslash(\Omega_0\cap U)$ is compatible with the flow $\Phi_t$ of $H$ on $U$, i.e.
\begin{equation}\label{eq:Phi_red_commutes}
 \tilde \Phi_t\circ Pr=Pr \circ\Phi_t
\end{equation}
\end{thm}
\begin{proof}
 see \cite[Theorem 8.1.1, p. 302]{OR04}
\end{proof}
 An important example for such an $G$-invariant open subset is the open and dense set of all principal orbits $M(H)$ described by Proposition \ref{prop:principalOrbit}. We will call the intersection of $\Omega_0$ with this set the \emph{regular part} of $\Omega_0$ and write $\tu{Reg }\Omega_0:=\Omega_0\cap M(H)$.

If the condition of equal $G$-orbits on $\Omega_0\cap U$ is not fulfilled, a reduction to a singular stratified space is possible \cite[Theorem 8.1.1, p. 302]{OR04}. The reduction is however much more complicated and uses results on the stratification of $G$-spaces in orbit types.

On the quantum mechanical side, the symmetry reduction appears via the Peter-Weyl theorem. The orthogonal action of $G$ on $\R^n$ induces a unitary representation on $L^2(\R^n)$, which is the space on which we study the quantized Hamilton operator, by
\[
 (M(g)\Psi)(x):=\Psi(g^{-1}x)~~\forall ~\Psi\in L^2(\R^n).
\]
By the Peter-Weyl theorem, this unitary representation leads to the following decomposition of our Hilbert space $L^2(\R^n)$: Let $\chi\in \hat G$ be an equivalence class of unitary irreducible $G$-representations given by an irreducible character of the compact group $G$ and $d_\chi=\chi(Id)$ the dimension of the irreducible representation. Then
\[
 P_\chi:=d\chi\int_G\overline{\chi(g)} M(g)dg ~\tu{ ($dg$ is the normalized Haar measure)}
\]
is a projector, and we denote its image by $L^2_\chi(\R^n)$. The space $L^2(\R^n)$ then decomposes into a direct sum
\[
 L^2(\R^n)=\bigoplus\limits_{\chi\in \hat G}L^2_\chi(\R^n).
\]
As a direct consequence of the $G$-invariance of the Hamiltonian $H$, we get that the quantized Hamilton operator $\hat H$ commutes with the $G$-action $M(g)$ so that
\begin{equation}\label{eq:HamiltonKommutiert}
 \hat H = M(g)\hat H M(g^{-1}).
\end{equation}
Equation (\ref{eq:HamiltonKommutiert}) implies that the Hamilton operator leaves $L^2_\chi(\R^n)$ invariant. We can therefore consider the restriction of $\hat H$ to $L^2_\chi(\R^n)$ and denote it by  $\hat H_\chi$. This restricted operator is the symmetry-reduced Hamilton operator. Note that that in contrast to the symmetry reduction of the classical system, for the quantum reduction there are no additional complications if $G$ acts with different orbit types. 

\section{Weak equivariant Weyl asymptotics}\label{sec:weak}
In this section we will prove a theorem for the asymptotic number of eigenvalues of the reduced Hamilton operator $\hat H_\chi$ for an arbitrary compact orthogonal group action. Suppose that $\hat H_\chi$ has purely discrete spectrum on $]E_1,E_2[$ and let $f\in C^\infty_0(]E_1,E_2[)$. Note that as $\hat H=\textup{Op}^w_ h(H)$ is a family of operators parametrized by $ h>0$, its spectrum will equally depend on the semiclassical parameter $ h$. We will give the leading term in the limit $ h\to0$, plus a remainder estimate, for the expression $\sum f(\lambda_i)$, where the sum is over all eigenvalues of $\hat H_\chi$ in $]E_1,E_2[$ repeated according to their multiplicity. By taking a smooth function approximating the characteristic function of an interval $[a,b]\subset]E_1,E_2[$, this expression immediately gives us approximate information on the asymptotic number of eigenvalues in the interval $[a,b]$. As we are restricted to smooth functions $f$, and cannot take $f$ exactly to be a 
characteristic 
function, we 
call these asymptotics ``weak Weyl asymptotics''. The strong version of the equivariant Weyl's law i.e.~asymptotics for $f$ being a characteristic function, have been obtained by El Houakmi-Helffer \cite{EH91} via FIOs in the case of one orbit type. 

We will prove the ``weak Weyl asymptotics'' for an arbitrary orbit structure by reducing the problem via functional calculus to oscillatory integrals with singular critical sets, and use results of \cite{Ram14} for these integrals. A generalization of the strong version of Weyl's law to arbitrary orbit structures seems not possible by reducing the problem to the integrals studied in \cite{Ram14}, but would require to resolve a much more complicated phase function.

We can now state and prove our first main result
\begin{thm}\label{thm:weakAsym}
Let $G$ be a compact subgroup of $O(n)$ and $H:\R^{2n}\to \R$ a smooth $G$-invariant Hamiltonian which is bounded from below, and satisfies the regularity conditions of Hypothesis \ref{hyp:regularityH}. Let $E_1<E_2$ and $\epsilon>0$ be such that $H^{-1}([E_1-\epsilon,E_2+\epsilon])$ is compact, and let $f:\R\to\R$ be smooth and compactly supported in $]E_1,E_2[$. For a given character $\chi\in \hat G$ and small $ h$, $f(\hat H_\chi)$ is of trace class, and its trace is asymptotically given by
\begin{equation}\label{eq:weakAsym}
 \tu{Tr}(f(\hat H_\chi)) = (2\pi  h)^{-n+\kappa} d_\chi\mathcal L_0 +\mathcal O( h^{-n+\kappa+1}\log( h^{-1})^\Lambda)
\end{equation}
where $\kappa$ is the dimension of a principal orbit in $\R^n$, $\Lambda$ the maximal number of elements of a totally ordered subset of the set of orbit types, and
\begin{equation}\label{eq:Leading_Term_weak_asym}
 \mathcal L_0 = [\pi_{\chi|H_{\mathrm{prin}}}:1]\int\limits_{\tu{Reg } \Omega_0} \frac{f(H(z))}{\tu{Vol}(Gz)}d(\tu{Reg }\Omega_0)(z).
\end{equation}
Here $\tu{Vol}(Gz)$ denotes the induced volume of the principal $G$-orbit through $z$ as a submanifold of $\R^n$, and $d(\tu{Reg }\Omega_0)$ the induced measure on $\tu{Reg }\Omega_0$. $[\pi_{\chi|H_{\mathrm{prin}}}:1]$ is the Frobenius factor which is given by the multiplicity with which the trivial representation occurs in the restriction of the irreducible $G$-representation $\pi_\chi$ to the isotropy subgroup of the principal orbits $H_{\mathrm{prin}}\subset G$.
\end{thm}
\begin{rem}
 As mentioned in Section \ref{sec:group_action}, the case of a compact symmetry group $G\subset GL(n,\R)$ can be reduced to the case of an orthogonal subgroup, treated in the theorem, by a standard averaging argument. 
\end{rem}
\begin{rem}
 As $G$ acts with only one orbit type $G/H_{\mathrm{prin}}$ on $\tu{Reg }\Omega_0$, the quotient $G\backslash\textup{Reg }\Omega_0 $ has the structure of a symplectic manifold (see Theorem \ref{thm:SymplRed}) and $\tu{Reg }\Omega_0$ is a fiber bundle over this quotient
 \[
  Pr:\Abb{\tu{Reg }\Omega_0}{G\backslash\textup{Reg }\Omega_0}{z}{[z]}
 \]
 with the principal orbit $G/H_{\mathrm{prin}}$ as the fiber. By its $G$-invariance, the Hamiltonian $H(z)$ takes  constant values along the fibers equal to the value of the reduced Hamiltonian $\tilde H([z])$. Consequently, using a fiberwise integration we can rewrite (\ref{eq:Leading_Term_weak_asym}) as
 \begin{equation}
  \mathcal L_0 = [\pi_{\chi|H_{\mathrm{prin}}}:1]\int\limits_{G\backslash \tu{Reg } \Omega_0} f(\tilde H([z]))d(G\backslash \tu{Reg }\Omega_0)([z]).
 \end{equation}

\end{rem}
\begin{proof}
In a first step we show that $f(\hat H_\chi)$ is trace class, and express the trace as an oscillatory integral. Afterwards we will derive the asymptotic behavior for this integral from a result in \cite{Ram14}.

From the compactness of $H^{-1}(E_1-\epsilon,E_2+\epsilon)$ we conclude that $\hat H$ has purely discrete spectrum on the interval $[E_1,E_2]$ with eigenvalues $\lambda_1,\dots,\lambda_r$ and eigenstates $\Psi_1,\dots\Psi_r$ \cite[Proposition 5.1]{HR83}. The expression $f(\hat H)$ is thus defined by
\[
 f(\hat H)\Psi = \sum\limits_{i=1}^rf(\lambda_i) \langle \Psi_i,\Psi\rangle \Psi_i, \tu{ for }\Psi\in L^2(\R^n).
\]
As $H(z)$ is $G$-invariant, $\hat H$ commutes with the left regular representation of $G$ and the eigenstate $\Psi_i$ can be chosen such that each $\Psi_i$ belongs precisely to one $G$-representation $L^2_\chi$. Consequently, if we denote by $\sigma_\chi$ the set of all eigenvalues belonging to $\chi\in \hat G$ we can write
\[
 f(\hat H)\Psi=\sum\limits_{\chi\in \hat G}\sum\limits_{\lambda_i\in \sigma_\chi} f(\lambda_i)\langle \Psi_i,\Psi\rangle \Psi_i=\left(\sum\limits_{\chi\in \hat G}f(\hat H_\chi)\right)\Psi
\]
with 
\[
f(\hat H_\chi)\Psi=\sum\limits_{\lambda_i\in \sigma_\chi} f(\lambda_i)\langle \Psi_i,\Psi\rangle \Psi_i. 
\]
Thus we conclude
\begin{equation}\label{eq:f(Hchi)}
 f(\hat H_\chi)=f(\hat H)\hat P_\chi.
\end{equation}
Furthermore $f(\hat H_\chi)$ is of finite rank and thus trace class. 
Now, for $f(\hat H)$ we can use the Helffer-Robert functional calculus \cite[Théorème 4.1]{HR83} and get for each $N\in \N$ the expression
\begin{equation}\label{eq:FunctCalc}
 f(\hat H)= \sum\limits_{j=0}^N  h^j Op_ h^w(a_j) +  h^{N+1} \hat R_{N+1}( h),
\end{equation}
where the $a_j$ are symbols with $\tu{supp}(a_j)\subset H^{-1}(]E_1-\epsilon,E_2+\epsilon[)$ and $a_0=f\circ H$. Furthermore, from the proof of Théorème 5.1 in \cite{HR83} it follows that $\hat R_N$ is of trace class for $ h$ sufficiently small and that its trace norm is bounded by $\|\hat R_{N+1}\|_{\tu{Tr}}\leq C h^{-n}$. 

In order to express $f(\hat H_\chi)$ as an oscillatory integral we first prove the following
\begin{lem}\label{lem:TraceKernel}
 Let $Op_ h(a_j)$ be the standard quantization. Then $Op_ h(a_j)\hat P_\chi$ can be written as an integral operator with smooth kernel
 \[
  K(x,y) = (2\pi h)^{-n}d_\chi\int\limits_G\int\limits_{\R^n} a_j(x,\xi)\overline{\chi(g)}e^{\frac i h (x-g^{-1}y)\xi}  d\xi dg.
 \]
\end{lem}
\begin{proof}
For $\Psi\in C_0^\infty(\R^n)$, we calculate with Fubini
\begin{eqnarray*}
   \int\limits_{\R^n} K(x,y)\Psi(y)dy &=&  (2\pi h)^{-n}d_\chi \int\limits_{\R^n}\left[\int\limits_G\int\limits_{\R^n} a_j(x,\xi)\overline{\chi(g)}e^{\frac i h (x-g^{-1}y)\xi}  d\xi dg\right]\Psi(y) dy\\
   &=&(2\pi h)^{-n} \int\limits_{\R^n}\int\limits_{\R^n} a_j(x,\xi)e^{\frac i h (x-y)\xi}  \left[\int\limits_Gd_\chi\overline{\chi(g)} \Psi(gy)dg\right] d\xi dy\\
   &=&[Op_ h(a_j) P_\chi \Psi] (x),
\end{eqnarray*}
where we used the fact that all integrands are compactly supported.
\end{proof}
In order to calculate the desired trace by the above lemma we need to change the quantization in (\ref{eq:FunctCalc}).
\begin{lem}\label{lem:ChangeQuant} 
 If $b\in S(m)$ is a $G$-invariant, compactly supported symbol, then
 \begin{equation}\label{eq:change_quantization}
   Op_ h^w(b)=Op_ h\left(\sum\limits_{j=0}^N  h^j b_j(x,\xi) + h^{N+1}R_N(x,\xi)\right)   
 \end{equation}

 for every $N\in \N$, where $b_j\in S(m)$ are compactly supported and $G$-invariant. Furthermore, $b_0=b$ and $\|Op_ h(R_N)\|_{Tr}<C h^{-n}$.
\end{lem}
\begin{proof}
 By \cite[Theorem 4.13]{Zw12} one has (\ref{eq:change_quantization}) with $b_j=\frac{(-i h \langle D_x,D_\xi\rangle)^j}{j!}b$, so that the compactness property follows immediately. From the $G$-invariance of $b$ we conclude for arbitrary $g\in G$
 \[
 \begin{array}{rcl}
  \left(\frac{(-i h \langle D_x,D_\xi\rangle)^j}{j!}b \right)(x)&=&\left(\frac{(-i h \langle D_x,D_\xi\rangle)^j}{j!}(b\circ g)\right)(x)\\
  &=&\left(\frac{(-i h \langle gD_x,gD_\xi\rangle)^j}{j!}b\right)(gx)\\
  &=&\left(\frac{(-i h \langle D_x,D_\xi\rangle)^j}{j!}b\right)(gx)\\
  \end{array}
 \]
which proves the $G$-invariance of the $b_j$. Finally, as $b$ is compactly supported, $b$ and consequently also $R_N$ belong to $S(\langle z\rangle^{-k})$ for all $k$. Thus, $R_N$ is a Schwartz function and $\tu{Tr}[Op_ h(R_N)]=(2\pi h)^{-n}\iint R_N(x,\xi)dxd\xi$.
\end{proof}
Equation (\ref{eq:f(Hchi)}) and (\ref{eq:FunctCalc}), together with Lemma \ref{lem:TraceKernel} and Lemma \ref{lem:ChangeQuant} finally yield
\begin{cor}
 For each $N\in \N$ the trace of $f(\hat H_\chi)$ can be written as
 \begin{equation}\label{eq:WeakAsymOscInt}
  \tu{Tr}(f(\hat H_\chi)) = (2\pi h)^{-n} d_\chi \sum\limits_{j=0}^N h^j \int\limits_G \int\limits_{\R^n} \int\limits_{\R^n} \tilde a_j(x,\xi)\overline{\chi(g)}e^{\frac i h (x-gx)\xi}  dxd\xi dg +\mathcal O( h^{-n+N}),
 \end{equation}
 where the $\tilde a_j$ are compactly supported, $G$-invariant symbols and $\tilde a_0=f\circ H$.
\end{cor}

Each of the summands in (\ref{eq:WeakAsymOscInt}) is thus an oscillatory integral with phase function $\Phi(x,\xi)=(x-gx)\xi$. The critical set of such an oscillatory integral is defined by $\mathcal C:=\{(z,g)\in \R^{2n}\times G,\,d\Phi_{z,g}=0\}$. For this phase function a straightforward computation yields
\[
 \mathcal C =\{(z, g)\in \Omega_0\times G,~gz=z\}. 
\]

 If $G$ acts with different orbit types, this set is not smooth and one cannot obtain the asymptotics directly by the stationary phase theorem. However, it has been shown \cite{Ram14} that by iteratively blowing up the critical set, the leading order term for integrals of this type can be obtained together with a remainder estimate. As all symbols of the oscillatory integral are compactly supported we can use \cite[Theorem 9.3]{Ram14} and obtain in the limit $ h \to 0$
\[
 \tu{Tr}(f(\hat H_\chi)) = (2\pi  h)^{-n+\kappa} d_\chi\mathcal L_0 +\mathcal O( h^{-n+\kappa+1}\log( h^{-1})^\Lambda),
\]
where $\kappa$ is the dimension of a principal orbit in $\R^n$, $\Lambda$ the maximal number of elements of a totally ordered subset of the set of orbit types, and
\[
 \mathcal L_0 = \int\limits_{\tu{Reg } \mathcal C} \frac{\overline{\chi(g)}f(H(x,\xi))}{\left|\det\left( \tu{Hess }\Phi(x,\xi,g)_{|N_{x,\xi,g}\tu{Reg }\mathcal C}\right)\right|^{1/2}}d(\tu{Reg }\mathcal C)(x,\xi,g),
\]
where the set $\tu{Reg }\mathcal C:=\mathcal C\cap(\tu{Reg }\Omega_0\times G)$ is a smooth submanifold of $\R^{2n}\times G$ and $\tu{Hess }\Phi(x,\xi,g)_{|N_{(x,\xi,g)}\tu{Reg }\mathcal C}$ the restriction of the Hessian to the normal-bundle of $\tu{Reg }\mathcal C$.
Applying \cite[Lemma 7]{CR09} this expression can be simplified to
\[
 \mathcal L_0 = [\pi_{\chi|H_{\mathrm{prin}}}:1]\int\limits_{\tu{Reg } \Omega_0} \frac{f(H(z))}{\tu{Vol}(Gz)}d(\tu{Reg }\Omega_0)(z).
\]
This finally proves Theorem \ref{thm:weakAsym}.
\end{proof}
\section{Equivariant Gutzwiller formula}
While the previous section treated a result on the number of eigenvalues of the symmetry-reduced operator $\hat H_\chi$, this section will be dedicated to the correlations in the spectrum of $\hat H_\chi$ which will be described by an equivariant Gutzwiller trace formula. In Section \ref{sec:SpecDist} we will introduce the equivariant spectral distribution which is the well known quantity to study asymptotic spectral correlations. This spectral distribution will be written as an oscillatory integral using results of Cassanas \cite{Cas06}. The phase analysis and a discussion of the possible singularities of the critical set will be the subject of Section \ref{sec:CritSet}. We will see that, in general, the singularities can result from the group action as well as from the Hamilton dynamics, and we will review the assumptions which were imposed by Guillemin-Uribe \cite{GU90} and Cassanas \cite{Cas06} in order to avoid these singularities. In Section \ref{sec:DynaAssumptions} we introduce a generalization of 
the 
former assumptions which makes no special hypothesis on the group action anymore. Under these assumptions we finally prove the equivariant Gutzwiller trace formula in Section \ref{sec:ProofGutzwiller}. 

\subsection{Spectral distribution and oscillatory integrals}\label{sec:SpecDist}
Throughout this section we will assume that $H$ is a Hamiltonian fulfilling Hypothesis \ref{hyp:regularityH} and \ref{hyp:closed} and that $G\subset O(n)$ is a compact Lie group acting on $\R^n$. Furthermore we assume that $H$ is $G$-invariant. As discussed above, for an arbitrary equivalence class of irreducible unitary representation $\chi\in \hat G$ one can study the symmetry-reduced Hamilton operator $\hat H_\chi$ which has purely discrete spectrum in $[E_1,E_2]$ for $ h\in ]0, h_0]$. Let $\zeta\in C_0^\infty([E_1,E_2])$ be a smooth cut-off function and $ h \in]0, h_0]$. We can then define for each $E\in ]E_1,E_2[$ the spectral distribution $\rho^\chi_{E, h}\in \mathcal S'(\R)$ by its action on $f\in \mathcal S(\R)$ by setting
\begin{equation}\label{eq:SpecDist}
 \rho^\chi_{E, h}(f):= \tu{Tr}\left(\zeta(\hat H_\chi)f\left(\frac{E-\hat H_\chi}{ h}\right)\right). 
\end{equation}
This expression can be reformulated using the Fourier inversion formula, and one obtains
\begin{equation}\label{eq:SpecDistFT}
 \rho^\chi_{E, h}(f):= \frac{1}{2\pi} \tu{Tr}\left( \zeta(\hat H_\chi)\int\limits_\R e^{\frac{i}{ h} E t}e^{-\frac{i}{ h} \hat H_\chi t}\hat f(t) dt\right)
\end{equation}
with $\hat f(t)$ being the Fourier transform $\hat f(t):=\int\limits_\R e^{-iEt}f(E)dE$. 

In order to understand the significance of this spectral distribution, consider the following heuristics: Assume, that $\hat f(t)$ is compactly supported and that this support is very close to $t_0\neq 0$ such that it approximates the delta distribution at $t_0$. One then obtains
\[
 \rho^\chi_{E, h}(f)\approx \frac{1}{2\pi} e^{\frac{i}{ h} E t_0}\sum\limits_{E_j\in \sigma(\hat H_\chi)\cap[E_1,E_2]} \zeta(E_j)e^{-\frac{i}{ h}E_j t_0}
\]
where $\sigma(\hat H_\chi)$ denotes the spectrum of $\hat H_\chi$. From the Weyl law we know that the number of eigenvalues $E_j$ of $\hat H_\chi$ in $[E_1,E_2]$ diverges as $ h\to0$. Therefore, if they are randomly distributed in $[E_1,E_2]$, the spectral distribution will be very small for small $ h$ due to phase cancellation. The contrary  extreme would be a totally correlated spectrum, respectively an equidistant spectrum
\[
E_{j+1}=E_j+\frac{2\pi h}{t_0} 
\]
In this case the spectral distribution gives 
\[
 \rho^\chi_{E, h}(f)\approx \frac{1}{2\pi} e^{\frac{i}{ h} (E-E_0) t_0}\sum\limits_{E_j\in \sigma(\hat H_\chi)\cap[E_1,E_2]} \zeta(E_j),
\]
which constitutes for small $ h$ a large quantity according to the weak asymptotics. Of course such an equidistant spectrum is only possible for very special dimensions of phase space and $G$-orbits in order to be in accordance with the equivariant Weyl law, and also the approximation of the delta distribution by  $\hat f$ strictly speaking does not hold anymore in the limit $ h\to 0$. This heuristic discussion however shows that contributions of the spectral distribution at $t_0\neq 0$ measure the correlation of the spectrum with a correlation length $\tfrac{2\pi  h}{t_0}$, motivating its study. The equivariant Gutzwiller trace formula will be an asymptotic expansion of the spectral distribution $\rho^\chi_{E, h}$ evaluated at a Schwartz function $f$ with $\hat f$ compactly supported away from zero.

The standard way to obtain the asymptotic expansion of the spectral distribution is to rewrite it as an oscillatory integral, and apply the stationary phase approximation. This transformation of the spectral distribution to an oscillatory integral can be performed in different ways. The most common way is to use Fourier integral operator theory in order to study the operator $e^{\frac{i}{ h} \hat H t}$ in (\ref{eq:SpecDistFT}). An alternative approach has been proposed by Combescure-Ralston-Robert \cite{CRR99} using their previous results on the propagation of coherent states \cite{CR97}. This coherent state approach has been elaborated in detail by Cassanas \cite{Cas06fg, Cas06} in exactly the same setting which we study in this work, i.e.~for symmetry-reduced $ h$-PDOs. We will thus use the following result of those works:
\begin{prop}\label{prop:OscIntGutzwiller}
Let $f\in \mathcal S(\R)$ such that $\hat f$ is compactly supported. Then
 \begin{equation}\label{eq:OscIntegralEquiv} 
  \rho_{E, h}^{\chi}(f)=d_\chi \int_\R dt\int_{\R^{2n}}dz\int_{G}dg~~a_g(z, h)\hat f(t)e^{\frac{i}{ h}\phi_E(z,t,g)}, 
 \end{equation}
where the complex phase function is given by $\phi_E=\phi_E^1+i\phi_E^2$ with
\begin{eqnarray*}
 \phi_E^1(z,t,g)&=&(E-H(z))t+\frac{1}{2}\langle g^{-1} z,Jz\rangle-\frac{1}{2}\int_0^t\langle(z_t-g^{-1}z),J\dot{z}_s\rangle ds,\\  
 \phi_E^2(z,t,g)&=&\frac{1}{4}\langle(\mathds{1}-\hat W_t)(gz_t-z),(gz_t-z)\rangle.
\end{eqnarray*}
Here $\hat W_t$ is a complex valued $2n\times2n$-matrix such that $(\mathds{1}-\hat W_t)$ defines a non-degenerate quadratic form. More precisely, $\hat W_t:=\left(\begin{array}{cc} W_t&-iW_t\\-iW_t&-W_t\end{array}\right)$ where the $n\times n$ matrice $W_t$ is given by the equation 
\[
\tfrac{1}{2}(\mathds{1}+W_t):=\left(\mathds{1}-ig(C+iD)(A+iB)^{-1}g^{-1}\right)^{-1}
\]
with $A,B,C,D$ being the real $t$- and $z$-dependent  $n\times n$ matrices such that the linearized flow is given by $(\Phi_t)_{*,z}=\left(\begin{array}{cc} A&B\\C&D\end{array}\right)$. The integrand $a_g(\bullet,  h):\R^{2n}\to \C$ is supported in $H^{-1}([E-\delta E,E+\delta E])$ and is given by an asymptotic expansion in $ h$ with leading term
\begin{equation}\label{eq:SymbolLeadingTerm}
 a_g(z, h)\sim_{ h\to 0^+}\frac{(2\pi h)^{-d}}{2\pi}\overline{\chi(g)}\chi_2(z)\zeta(H(z))\tu{det}_*^{-1/2}\left(\frac{A+iB-i(C+iD)}{2}\right).
\end{equation}
Here $\chi_2$ is a smooth cut-off function compactly supported around $\Sigma_E$ and equal to one in a neighborhood of $\Sigma_E=H^{-1}(E)$ and $\tu{det}_*^{-1/2}(M)$ is defined as the product of the reciprocal square roots of eigenvalues of the matrix $M$ with real parts chosen to be positive. 
\end{prop}
\begin{proof}
See \cite{Cas06}, Section 4, Equation (4.1). For the non degeneracy of $\mathds{1}-\hat W_t$ see the discussion in \cite{Cas06fg} on page 10 (Section IV A). 
\end{proof}
By this proposition we have thus written the spectral distribution as an oscillatory integral with complex phase function. The next section will be dedicated to its phase analysis.
\subsection{The critical set}\label{sec:CritSet}
For a complex valued phase function $\Phi\in C^\infty(\R^n)$ with $\tu{Im}(\Phi)\geq0$, the critical set is defined as (c.f. \cite{HoeI} Theorem 7.7.1)
\[
 \mathcal C_\Phi :=\{x\in \R^n,~\tu{Im}(\Phi(x))=0 \tu{ and } \nabla\Phi(x)=0\}.
\]
Using the fact that $(\mathds{1}-\hat W_t)$ is non-degenerate, a straightforward but slightly tedious calculation leads to
\begin{prop}\label{prop:CritSetGutzwiller}
 For the complex phase function $\phi_E$ of Proposition \ref{prop:OscIntGutzwiller}, the critical set is given by
\begin{equation}\label{eq:CritSetEquivariant}
 \mathcal C_{\phi_E}=\{(z,t,g)\in \R^{2n}\times\R\times G,~~z\in\Omega_0\cap\Sigma_E,~g\Phi_t(z)=z\}
\end{equation}
\end{prop}
\begin{proof}
 See \cite[Proposition 4.1]{Cas06}.
\end{proof}
This critical set $ \mathcal C_{\phi_E}$ is in general a singular set. There are two different sources for these singularities: First of all, they can purely result from the Hamilton flow even if there is no symmetry group acting. For example, an isolated fixed point or a singular family of periodic orbits would lead to such singularities. A second source of singularities are, as in the case of the weak asymptotics, different orbit types of the group action. The same critical set with the same sources of singularities also appears in the work of Guillemin-Uribe \cite{GU90}. In order to make sure that the critical set is smooth, and to prove the asymptotic expansion for the equivariant Gutzwiller trace formula, both works \cite{GU90, Cas06} impose the following three conditions \ref{hyp:Reduction}-\ref{hyp:NondegOrbits} which we now recall, using the notation of \cite{Cas06}.
\begin{customhyp}{(H1)}\label{hyp:Reduction}
 For $E\in \R$, the \emph{hypothesis of reduction} is fulfilled if there exists some $\delta E>0$ such that for $U=H^{-1}(]E-\delta E,E+\delta E[)$
 \begin{itemize}
  \item $\Omega_0\cap U\neq \emptyset$, 
  \item all $G$-orbits in  $\Omega_0\cap U$ are of the same type.
 \end{itemize}
\end{customhyp}
This hypothesis prevents the possible singularities coming from the group action. If it is fulfilled one can consider the reduced Hamiltonian $\tilde H:G\backslash(\Omega_0\cap U)\to\R$ given by Theorem \ref{thm:SymplRed}. In order to avoid the singularities coming from the Hamilton dynamics, one further needs the following two Hypothesis.
\begin{customhyp}{(H2)}\label{hyp:non_critical} 
Assume that \ref{hyp:Reduction} is fulfilled. Then $E$ is a \emph{non critical value} of $\tilde H$, i.e.\,$d\tilde H\neq 0$ on $\tilde \Sigma_E:=\tilde H^{-1}(E)$.
\end{customhyp}
\begin{customhyp}{(H3)}\label{hyp:NondegOrbits}
 Assume that \ref{hyp:Reduction} and \ref{hyp:non_critical} are fulfilled. Let $\gamma\subset \tilde \Sigma_E$ be an arbitrary closed orbit of the reduced Hamiltonian $\tilde H$, $[z]\in \gamma$ an arbitrary point on this orbit, and $T_\gamma>0$ the period of the orbit so that $\tilde \Phi_{T_\gamma}([z])=[z]$. Let furthermore $\tilde F_t:=\partial_z\tilde \Phi_t$ be the linearized flow. Then  we assume that $\gamma$ is a \emph{non-degenerate orbit} i.e.~the eigenvalue 1 of $\tilde F_{T_\gamma}([z])$ has algebraic multiplicity 2.
\end{customhyp}
More generally, the non-degenerate orbit condition can be weakened replacing it by the \emph{clean-flow condition}. However, under this assumption, the resulting spectral asymptotics can not be written as a sum over periodic orbits anymore, and are much less explicit. 

In the following we shall drop the Hypothesis \ref{hyp:Reduction} and study the Gutzwiller trace formula for general compact group actions. As the formulation of \ref{hyp:non_critical} and \ref{hyp:NondegOrbits} is usually based on \ref{hyp:Reduction}, these two conditions have to be replaced by appropriate generalizations which will be formulated in the next subsection.

\subsection{The $G$-non-stationary and $G$-non-degenerate assumption}\label{sec:DynaAssumptions}
Let \ref{hyp:Reduction} be fulfilled. Then assumption \ref{hyp:non_critical} is equivalent to the fact that the Hamilton flow of $\tilde H$ on the reduced phase space has no stationary points at the given energy level $E$. Suppose that $[z]\in \tilde \Sigma_E$ is such a stationary point, i.e.\,$\tilde\Phi_{t}([z])=[z]$ for all $t$, then we infer from (\ref{eq:Phi_red_commutes}), that $ \Phi_t(z)\subset Pr^{-1}([z])=Gz\subset \Omega_0\cap\Sigma_E$ for all $t$. In particular we have $\frac{d}{dt}\Phi_t(z)_{|t=0} \in \g z$. If, on the other hand, we have for some $z\in \Omega_0\cap\Sigma_E$ that $\frac{d}{dt}\Phi_t(z)_{|t=0}= A z$ for some $A\in \g$, then we obtain that $\Phi_t(z)=e^{At}z$ because we check that it fulfills the Hamilton equations of motions (\ref{eq:HamiltonEquations})
\begin{eqnarray*}
 \frac{d}{dt}e^{At}z_{|t=t_0} &=& e^{At_0}Az\\
 &=&e^{At_0}\frac{d}{dt}\Phi_t(z)_{|t=0}\\
 &=&e^{At_0}J\nabla H(z)\\
 &=&J\nabla H(e^{At_0}z).
\end{eqnarray*}
Here we used in the last equality that the group action is symplectic and that the Hamiltonian is $G$-invariant. Consequently a stationary point on the symmetry-reduced system in $\tilde \Sigma_E$ is equivalent to an orbit of the Hamilton flow of $H$ which is tangent to a $G$-orbit in at least one point in $\Omega_0\cap\Sigma_E$. Let now \ref{hyp:Reduction} be dropped. Hypothesis \ref{hyp:non_critical} can then be easily reformulated for general compact group actions and we arrive at
\begin{customhyp}{(H2$'$)}\label{hyp:H2prime}
  The given  Hamilton dynamics is called \emph{$G$-non-stationary} for a given energy $E\in \R$ if and only if $J\nabla H(z)\nsubseteq \g z$ for all $z\in \Sigma_E\cap \Omega_0$.
\end{customhyp}

The generalization of \ref{hyp:NondegOrbits} requires a little bit more work. We assume that the Hamilton dynamics is $G$-non-stationary, so that in particular $\nabla H(z)\neq 0$ on $\Omega_0\cap\Sigma_E$ and $\Sigma_E$ is a smooth submanifold in some neighborhood of $\Omega_0$. From the $G$-invariance of the Hamiltonian we obtain that the $G$-orbit $Gz$ of a point $z\in \Omega_0\cap\Sigma_E$ is completely contained in $\Sigma_E$ thus the same inclusion holds for their tangent spaces $\g z\subset T_z\Sigma_E$. Consequently $\nabla H(z)\perp Az$ for all $z\in\Omega_0\cap\Sigma_E$ and $A\in \g$. For $z\in \Omega_0\cap\Sigma_E$, we obtain thus the decomposition
\begin{equation}\label{eq:DecompTangentSpaceEquiv}
T_z\R^{2n} =(\R \nabla H(z)\oplus J\g z)\operp(\R J\nabla H(z)\oplus\g z)\operp \mathcal R,
\end{equation}
where $\mathcal R$ is simply the orthogonal complement of $(\R \nabla H(z)\oplus J\g z)\operp(\R J\nabla H(z)\oplus\g z)$ in $T_z\R^{2n}$. In order to obtain a clearer interpretation of (\ref{eq:DecompTangentSpaceEquiv}) we introduce a notation which will turn out to be convenient at several other points in the sequel. We define the Lie group $F:=\R\times G$ and its action for $f=(t,g)\in F$ on $z\in \R^{2n}$ by
 \[
 fz:=g\Phi_t(z).
 \]
 The fact that the $G$-action commutes with $\Phi_t$ assures that this is really a Lie group action. We will denote its Lie algebra by $\f=\R\oplus\g$ and an Lie-algebra element by $(\tau, A)$. Its infinitesimal Lie-algebra action is given by $(\tau, A)z= \tau J\nabla H(z)+Az$ and for the subspace of $T_z\R^{2n}$ spanned by the $\f$-action we write
 \[
  \f z:=\{\tau J\nabla H(z)+Az,~(\tau, A)\in \f\}=T_z(Fz).
 \]
With this notation (\ref{eq:DecompTangentSpaceEquiv}) becomes
\begin{equation}\label{eq:DecompTangentSpaceF}
T_z\R^{2n} =J\f z\operp\f z\operp \mathcal R.
\end{equation}

Consider now a point $z\in \Omega_0\cap\Sigma_E$ which fulfills $fz=z$ for some $f\in F$. We can then study the differential $f_{*,z}=(g\Phi_T)_{*,z}:T_z\R^{2n}\to T_z\R^{2n}$.
\begin{lem}\label{lem:GeneralFormDgPhi}
 Assume that \ref{hyp:H2prime} is fulfilled and consider for a fixed $f=(T,g)\in F$ a point $z\in \Omega_0\cap\Sigma_E$ such that $fz=g\Phi_T(z)=z$. Then the differential of $f$ has the form
\begin{equation}\label{eq:GeneralFormDgPhi}
 f_{*,z}=\left(\begin{array}{ccc}
             \begin{array}{|c|} \hline \mathcal A\\ \hline \end{array} &0&0\\
             *&\begin{array}{|c|} \hline \mathcal B\\ \hline \end{array}&*\\
             *&0&\begin{array}{|c|} \hline \mathcal P\\ \hline \end{array}
\end{array}
\right)
\end{equation}
where the block form is with respect to the decomposition (\ref{eq:DecompTangentSpaceF}) of the tangent space, i.e.~$\mathcal A$ acts on $J\f z$, $\mathcal  B$ on $\f z$, and $\mathcal P$ on $\mathcal R$. Furthermore, $\mathcal B$ is given on $\f z=\R J\nabla H(z)\oplus \g z$ by
\[
 \mathcal B=\left(\begin{array}{cc}
          1&0\\
          0&\tu{Ad}(g)
         \end{array}
\right)
\]
\end{lem}
\begin{proof}
 As the $F$-action leaves the corresponding $F$-orbits invariant, $\f z$ is invariant under $f_{*,z}$ which explains the second column in (\ref{eq:GeneralFormDgPhi}). As the Hamilton flow and the $G$ action are both symplectic, the $F$ action is symplectic as well, which implies that the symplectic complement $(J\f z)^\perp=\f z \operp \mathcal R$ of $\f z$ is an invariant subspace under $f_{*,z}$ and consequently we obtain the first row in (\ref{eq:GeneralFormDgPhi}). For the refined form of $\mathcal B$ we consider the vector $J\nabla H(z)=\tfrac{d}{dt}\Phi_t(z)_{|t=0}$. Since
\[
 \frac{d}{dt}g\Phi_T(\Phi_t(z))_{|t=0} = \frac{d}{dt}g\Phi_t(\Phi_T(z))_{|t=0}= \frac{d}{dt}\Phi_t(g\Phi_T(z))_{|t=0} =J\nabla H(z),
\]
$J\nabla H(z)$ is eigenvector of $(g\Phi_T)_{*,z}$ with eigenvalue $1$. On the other hand, if $v\in \g z$ i.e.~$v=\tfrac{d}{dt}e^{At}z_{|t=0}$, one computes
\[
\left(g\Phi_T\right)_{*,z}(v)=\frac{d}{dt}g\Phi_T(e^{At} z)_{|t=0} =\frac{d}{dt}ge^{At}g^{-1} g\Phi_T( z)_{|t=0} =(\tu{Ad}(g)A) z.
\]
\end{proof}
With this general form of the differential $(g\Phi_T)_*$ one can now reformulate \ref{hyp:NondegOrbits} to a new assumption \ref{hyp:GNonDeg} which we call \emph{$G$-non-degenerate assumption}. In order to formulate it we first introduce some further notations.
\begin{Def}
 By a \emph{relative periodic point} $z\in \R^{2n}$ we will denote a point for which there are $T\neq 0$ and $g\in G$ such that $g\Phi_T(z)=z$. For a relative periodic point $z_0$ the set $\gamma=\{g\Phi_t(z_0),~g\in G\textup{ and }t\in \R\}$ is then called \emph{relative periodic orbit}. For a given relative periodic orbit $\gamma$, its \emph{primitive period length} $T_\gamma$ is defined as the smallest $T>0$ such that there exists $g\in G$ with $g\Phi_T(z) =z$ for some $z\in \gamma$. 
\end{Def}
\begin{rem}
 Relative periodic points are thus simply fixed points under some element $f=(T,g)\in F$ with $T\neq 0$ and the relative periodic orbits are the $F$-orbits of a relative periodic point.
\end{rem}
\begin{Def}\label{def:GNonDeg}
 Let $z\in \Omega_0\cap\Sigma_E$ be a relative periodic point with $g\Phi_{T}(z)=z$ for some $g\in G$ and $T>0$. Then this point, as well as the corresponding relative periodic orbit $\gamma$, are called \emph{$G$-non-degenerate} if $\mathcal P$ in (\ref{eq:GeneralFormDgPhi}) has no eigenvalue equal to 1.
\end{Def}

\begin{customhyp}{(H3$'$)}
\label{hyp:GNonDeg}
 We assume that all relative periodic points in $\Omega_0\cap\Sigma_E$ are \emph{$G$-non-degenerate}.  
\end{customhyp}
This condition assumes assumption \ref{hyp:H2prime} but not assumption \ref{hyp:Reduction}. If however the condition of reduction \ref{hyp:Reduction} is fulfilled the following equivalence holds.
\begin{lem}
 If the hypothesis of reduction \ref{hyp:Reduction} is fulfilled, then an orbit is $G$-non-degenerate if and only if this orbit is non-degenerate in the sense of \ref{hyp:NondegOrbits}.
\end{lem}
\begin{proof}

 If the condition of reduction is satisfied for $E\in \R$ we obtain a smooth symplectic manifold $\Omega_{\tu{red}} =G \backslash (\Omega_0\cap U)$ with a reduced Hamiltonian $\tilde H([z])=H(z)$ and the reduced Hamilton flow given by $\tilde \Phi_t([z]):=[\Phi_t(z)]$ for $[z]\in \Omega_{\tu{red}}$. The condition that $[z]$ is a periodic point under $\tilde \Phi_t$ is equivalent to the fact that $z$ is a relative periodic point, i.e.~that there is a $g\in G$ and $T> 0$ such that $g\Phi_T(z)=z$.

 As all Hamilton flows preserve the energy, the Hamilton vector field $\tilde H$ of the reduced Hamiltonian is tangential to $\tilde \Sigma_E = \tilde H^{-1}(E)$. We therefore have 
\begin{equation}\label{eq:decomp_tang_red}
 T_{[z]}(\Omega_{\tu{red}})=X_{\tilde H}([z])\operp \nabla \tilde H([z])\operp\mathcal R',~ [z]\in \tilde \Sigma_E.
\end{equation}
As the direction of the flow is always an eigenvector of the linearized flow with eigenvalue 1, we obtain for a relative periodic point $[z]$ the block form
 \begin{equation}\label{eq:differential_red_flow}
  (\tilde \Phi_T)_{*,[z]}=\left(\begin{array}{ccc}
             1&*&*\\
             0&1&0\\
             0&*&\begin{array}{|c|} \hline P\\ \hline \end{array}
\end{array}
\right)
 \end{equation}
with respect to the decomposition (\ref{eq:decomp_tang_red}). This block form immediately implies that the characteristic polynomial of the linearized flow is
\[
 \det\left(\lambda \mathds 1 - (\tilde\Phi_T)_{*,[z]}\right) = (\lambda-1)^2\det(\lambda \mathds 1-P).
\]
The ordinary hypothesis of non-degenerate orbits \ref{hyp:NondegOrbits}, demanding that the algebraic multiplicity of 1 is at most 2, is thus equivalent to the fact that $P$ has no eigenvalue equal to 1.

We finally show, that $P$ corresponds to $\mathcal P$ in (\ref{eq:GeneralFormDgPhi}). In order to see this we first consider the restriction of the reduced flow to the reduced energy shell $\tilde \Sigma_E=\tilde H^{-1}(0)$. As hypothesis \ref{hyp:H2prime} implies that $\nabla \tilde H([z])\neq 0$ for all $[z]\in \tilde \Sigma_E$ the reduced energy shell is again a smooth manifold. Equation (\ref{eq:differential_red_flow}) then reduces to
 \[
  ((\tilde \Phi_T)_{|\tilde \Sigma_E})_{*,[z]}=\left(\begin{array}{cc}
             
             1&*\\
             0&\begin{array}{|c|} \hline P\\ \hline \end{array}
\end{array}
\right)
 \]
 where the decomposition in block form is with respect to the decomposition 
\begin{equation}\label{eq:Decomp_OmegaRed_cap_Sigma}
T_{[z]}(\tilde \Sigma_E)=X_{\tilde H}([z])\operp \mathcal R'.
\end{equation}
As \ref{hyp:GNonDeg} implies \ref{hyp:H2prime} which is under \ref{hyp:Reduction} equivalent to \ref{hyp:non_critical} we conclude from \cite[Lemma 4.5]{Cas06} that $\Omega_0 \cap \Sigma_E$ is again a smooth manifold. We can therefor consider a similar restriction of $g\Phi_T$to $\Omega_0 \cap \Sigma_E$. Using the fact that $T_z(\Omega_0\cap\Sigma_E)= (J\f z)^\perp$ we obtain from Lemma \ref{lem:GeneralFormDgPhi} the expression
 \[
((g\Phi_T)_{|\Omega_0\cap\Sigma_E})_{*,z}=\left(\begin{array}{ccc}
             1 &0&*\\
             0&\begin{array}{|c|} \hline Ad(g)\\ \hline \end{array}&*\\
             0&0&\begin{array}{|c|} \hline P'\\ \hline \end{array}
\end{array}
\right),
\]
 where the decomposition in block form is with respect to the decomposition 
\begin{equation}\label{eq:Decomp_Omega_cap_Sigma}
T_{z}(\Omega_0\cap \Sigma_E)=\left(J\nabla H(z)\oplus \g z\right)\operp \mathcal R.
\end{equation}
From the compatibility of the reduced flow with the projection 
\[
Pr: \Omega_0\cap\Sigma_E\to G\backslash(\Omega_0\cap\Sigma_E)=\tilde\Sigma_E
\]
we get
\[
 Pr\circ ((g\Phi_T)_{|\Omega_0\cap\Sigma_E})=((\tilde \Phi_T)_{|\tilde \Sigma_E})\circ Pr
\]
and consequently
\begin{equation}\label{eq:Pr_commutes_restricted}
 Pr_{*,z}\circ((g\Phi_T)_{|\Omega_0\cap\Sigma_E})_{*,z}=((\tilde \Phi_T)_{|\tilde \Sigma_E})_{*,[z]}\circ Pr_{*,z}.
\end{equation}
With respect to the decomposition (\ref{eq:Decomp_OmegaRed_cap_Sigma}) and (\ref{eq:Decomp_Omega_cap_Sigma}) the differential of the projection $Pr_{*,z}:T_z(\Omega_0\cap \Sigma_E)\to T_{[z]}\tilde \Sigma_E$ can be written as
\[
 Pr_{*,z}=\left(\begin{array}{ccc}
             1 &0&0\\
             0 &0&\mathds{1}
\end{array}
\right).
\]
Inserting this expression in (\ref{eq:Pr_commutes_restricted}) and comparing both sides finally yields $P=\mathcal P$ and finishes the proof of this Lemma. 
\end{proof}
The following simple example shows that in systems with different orbit types it is quite likely that the $G$-non-degenerate assumption holds, and that there is a large class of systems where \ref{hyp:H2prime} and \ref{hyp:GNonDeg} hold, but not (H1), (H2) and (H3). 
\begin{exmpl}\label{exmpl:SingFromDiffGOrbits}
Let us consider the Hamiltonian of a 3-dimensional harmonic oscillator with two different frequencies $\omega_1=2\pi$ and $\omega_2=\frac{2\pi}{\sqrt 2}$
\[
 H(x,\xi)=\frac{1}{2}\left((2\pi)^2\cdot(x_1^2+x_2^2)+\left(\frac{2\pi}{\sqrt{2}}\right)^2x_3^2+|\xi|^2\right),~~x,\xi\in \R^3.
\]
The group $SO(2)$ acts on the phase space $\R^6$ symplectically as a symmetry group by acting canonically on the variables $x_1, x_2$ and $\xi_1,\xi_2$, respectively. As this action stabilizes the points $(0,0,x_3,0,0,\xi_3)$ it acts with different orbit types on $\Sigma_E$ for any $E>0$. In particular \ref{hyp:Reduction} is not satisfied. The zero level of the momentum map consists of the set of points with zero angular momentum in $x_3$ direction
\[
 \Omega_0=\{(x_1,x_2,x_3,\lambda\cdot x_1,\lambda\cdot x_2,\xi_3):~x_i,\xi_3,\lambda\in \R\}.
\]
The general solutions $(x(t),\xi(t))=\Phi_t(x_1,x_2,x_3,\xi_1,\xi_2,\xi_3)$ are explicitly given by
\begin{eqnarray*}
 x(t)&=&\left(x_1\cos(2\pi t)+\frac{\xi_1}{2\pi}\sin(2\pi t),~ x_2\cos(2\pi t)+\frac{\xi_2}{2\pi}\sin(2\pi t), ~x_3\cos\left(\frac{2\pi}{\sqrt{2}} t\right)+\right.\\
 &&\left.\frac{\sqrt{2}\xi_3}{2\pi}\sin\left(\frac{2\pi}{\sqrt{2}} t\right)\right),\\
 \xi(t)&=&\left(\xi_1\cos(2\pi t)-2\pi x_1\sin(2\pi t), ~ \xi_2\cos(2\pi t)-2\pi x_2\sin(2\pi t),~\xi_3 \cos\left(\frac{2\pi}{\sqrt{2}} t\right)-\right.\\
 &&\left. \frac{2\pi}{\sqrt{2}} x_3\sin\left(\frac{2\pi}{\sqrt{2}} t\right) \right).\\
\end{eqnarray*}
Thus, all points with $x_3=\xi_3=0$ are periodic with primitive period length 1, whereas all points with $x_1=x_2=\xi_1=\xi_2=0$ have primitive period length $\sqrt{2}$. All other points are not relative periodic at all, as the two frequencies have irrational ratio. Comparing 
\[
 J\nabla H(z)=\left(\xi_1,\xi_2,\xi_3,-(2\pi)^2 x_1,-(2\pi)^2 x_2,-\frac{(2\pi)^2}{\sqrt 2}\right)
\]
and
\[
 \g z= \R(x_2,-x_1,0,\xi_2,-\xi_1,0)
\]
we see that all points in $\Sigma_E\cap\Omega_0$ with $E>0$ are $G$-non-stationary thus condition \ref{hyp:H2prime} is fulfilled (note that there are however points $z\in\Sigma_E$ with $J\nabla H(z)\in \g z$).  Looking at the solutions $(x(t),\xi(t))$ one sees that the points $z\in\Omega_0\cap \Sigma_E \cap\{x_3=\xi_3=0\}$ have a smaller primitive period equal to $\tfrac{1}{2}$ as relative periodic orbits. Besides there is only one other relative periodic orbit, the periodic orbit with period $\sqrt 2$ in $\Omega_0\cap \Sigma_E \cap\{x_1=x_2=\xi_1=\xi_2=0\}$. 

Both relative periodic orbits are $G$-non-degenerate, which follows directly from the fact that the Hamilton flow $\Phi_t(x,\xi)$ is linear with respect to $(x,\xi)$ and the ratio of the frequencies $\omega_1$ and $\omega_2$ are irrational. To illustrate this, we take the relative periodic point $z_0=(1,0,0,0,0,0)\in\Omega_0\cap\Sigma_{2\pi^2}$ of the relative periodic orbit with primitive period $1/2$. It fulfills $g\Phi_{1/2}(z_0)=z_0$ for $g=-1\in SL(2,\R)$. The decomposition (\ref{eq:DecompTangentSpaceEquiv}) is given by
\[
 T_{z_0}\R^3=\tu{span}(e_1,e_5)\operp\tu{span}(e_2,e_4)\operp\mathcal R
\]
with $\mathcal R=\tu{span}(e_3,e_6)$ and $e_i$ being the canonical basis of $\R^6\cong T_{z_0}\R^6$. Then the restriction of the differential $(g\Phi_{k/2})_{*,z_0}$ to $\mathcal R\cong\R^2$ for $k\in \N_{>0}$ is given by
\[
 \mathcal P=((g\Phi_{k/2})_{*,z_0})_{|\mathcal R} =\left(\begin{array}{cc} 
                  \cos\left(\frac{2\pi k}{\sqrt{2}}\right)&\frac{\sqrt{2}}{2\pi}\sin\left(\frac{2\pi k}{\sqrt{2}}\right)\\
                  -\frac{2\pi}{\sqrt{2}}\sin\left(\frac{2\pi k}{\sqrt{2}}\right)&\cos\left(\frac{2\pi k}{\sqrt{2}}\right)                                               
                                               \end{array}
\right)
\]
This linear transformation has however no eigenvalue 1 for all $k \in \N_{>0}$. The arguments for the other relative periodic orbit are completely analogous.
\end{exmpl}
Similarly to the case of non-degenerate orbits, the hypothesis of $G$-non-degenerate orbits implies that the relative periodic orbits are discrete. This will be an important property for the proof of the equivariant Gutzwiller trace formula.
\begin{prop}\label{prop:DiscreteGNonDeg}
 Let $T>0$ and denote by $(\Gamma_E^{\tu{rel}})_T$ the set of all relative periodic orbits $\gamma$ in $\Omega_0\cap\Sigma_E$ with $0\leq|T_\gamma|\leq T$. If all $\gamma\in (\Gamma_E^{\tu{rel}})_T$ are non-degenerate, then they are discrete, i.e.~for each $z\in \gamma$ there is $\epsilon >0$ such that the only relative periodic points in $B_\epsilon(z)\cap\Omega_0\cap\Sigma_E$ with primitive period length smaller than $T$ are the points on the orbit $\gamma$.
\end{prop}
\begin{proof}
 Suppose that $(\Gamma_E^{\tu{rel}})_T$ is not discrete, so that there exists a relative periodic point $z_0\in \Omega_0\cap \Sigma_E$ with $g_0\Phi_{T_0}(z_0)=z_0$ and a sequence $z_n\in \Omega_0\cap\Sigma_E$, $z_n\to z_0$, of relative periodic points belonging to mutually disjoint relative periodic orbits which fulfill $g_n\Phi_{T_n}(z_n)=z_n$ with $|T_n|\leq T$. At the point $z=z_0$ we will now construct an eigenvector of the matrix $\mathcal P$ in (\ref{eq:GeneralFormDgPhi}) with eigenvalue 1 leading to a contradiction to the $G$-non-degenerate assumption.
 
 \emph{Step 1:} As $[-T,T]\times G$ is compact we can assure after going to a subsequence, that $f_n:=(T_n,g_n)\in [-T,T]\times G\subset F$ converges to the element $f_\infty=(T_\infty,g_\infty)\in [-T,T]\times G$. 
 
 \emph{Step 2:} From 
\begin{equation}\label{eq:lim_z0}
 z_0=\lim\limits_{n\to\infty} z_n=\lim\limits_{n\to\infty} f_n z_n = f_\infty z_0  
\end{equation}
we deduce that $z_0$ is relative periodic with period length $T_\infty$.
 
 \emph{Step 3:} By definition of the vector space $\g z_0$, for each normed vector $e_i\in \g z_0$ there is $A\in \g$ such that $\langle A{z_0}, e_i\rangle \neq 0$. From the continuity of the $G$-action it follows that there is a neighborhood $U\subset \R^{2n}$ of $z_0$ such that for all $z\in U$ we still have $\langle A{z}, e_i\rangle \neq 0$. From the continuity of $J\nabla H$ follows that $U$ can be chosen such that additionally $\langle J\nabla H(z), J\nabla H(z_0)\rangle \neq 0$. We can thus choose a neighborhood of $z_0$ such that all relative periodic orbits ($F$-orbits) intersect the affine vector space tangent to $(\f z_0)^\perp=(\R J\nabla H(z_0)\oplus \g z_0)^\perp\subset T_{z_0}$ transversally. Thus, for sufficiently large $n$ we can assume that $z_n$ is in such a neighborhood, and consequently we can choose a different point $\tilde z_n$ on the same relative periodic orbit with $\tilde z_n-z_0 \in (\R J\nabla H(z_0)\oplus \g z_0)^\perp$. Consequently we can suppose without loss of 
generality that our sequence of points $z_n$ fulfills $z_n-z_0\in (\R J\nabla H(z_0)\oplus \g z_0)^\perp$.
 
 \emph{Step 4:} Consider the sequence $\frac{z_n-z_0}{\|z_n-z_0\|}\in \mathcal S^{2n-1}$. After restricting once more to a subsequence, we can assume, that $\frac{z_n-z_0}{\|z_n-z_0\|}\to v$ with $\|v\|=1$. As $z_n\in \Omega_0\cap\Sigma_E$, we get $v\in (\R\nabla H(z_0)\oplus J\g z_0)^\perp$. From Step 3 we furthermore obtain 
\[
v\in \Big((\R\nabla H(z_0)\oplus J\g z_0)\operp(\R J\nabla H(z_0)\oplus \g z_0)\Big)^\perp. 
\]
So $v$ belongs to the subspace $\mathcal R$ of (\ref{eq:DecompTangentSpaceEquiv}) and the aim of the remaining steps in this proof is to show that it is an eigenvector of the matrix $\mathcal P$.

 \emph{Step 5:} Setting $t_n:=\|z_n-z_0\|$ we can choose a smooth curve $\gamma\subset \Sigma_E$ such that $\gamma(0)=z_0$ and $\gamma(t_n)=z_n$. For this curve we calculate $\dot \gamma(0)=\lim\limits_{n\to\infty}\frac{\gamma(t_n)-\gamma(0)}{t_n}=v$. 
 
 \emph{Step 6:}
 With this smooth path and (\ref{eq:lim_z0}) we calculate 
 \begin{eqnarray*}
  (f_\infty)_{*,z_0}(v)&=&\frac{d}{dt}f_\infty\gamma(t)_{|t=0}\\
  &=&\lim\limits_{n\to\infty}\frac{f_\infty\gamma(t_n)-f_\infty\gamma(0)}{t_n}\nonumber\\
  &=&\lim\limits_{n\to\infty}\frac{f_\infty z_n-f_n z_n+z_n-z_0}{t_n}\nonumber\\
  &=&v+\lim\limits_{n\to\infty}\frac{f_\infty z_n -f_n z_n}{t_n}.
 \end{eqnarray*}
 Now, using the fact that the $F$-action is smooth, Taylor expansion yields for any $f\in F$
 \[
  f z_n=fz_0+f_{*,z_0}(z_n-z_0)+O(t_n^2). 
 \]
 Consequently we obtain
 \begin{eqnarray*}
  (f_\infty)_{*,z_0}(v)&=&v+\lim\limits_{n\to\infty}\left[\frac{f_\infty z_0 -f_n z_0}{t_n} + ((f_\infty)_{*,z_0}-(f_n)_{*,z_0})\left(\frac{z_n-z_0}{t_n} \right)\right]\\
  &=&v+\lim\limits_{n\to\infty}\left[\frac{f_\infty z_0 -f_n z_0}{t_n}\right].
 \end{eqnarray*}
 Since obviously 
 \[
  \lim\limits_{n\to\infty}\left[\frac{f_\infty z_0 -f_n z_0}{t_n}\right]\in \f z_0,
 \]
 and $\mathcal P$ was the restriction of $(f_\infty)_{*,z_0}$ to $(J\f z_0\oplus \f z_0)^\perp$, $v$ is an eigenvector of $\mathcal P$ with eigenvalue 1 which finishes the proof.
 \end{proof}

\subsection{Proof of the equivariant Gutzwiller formula}\label{sec:ProofGutzwiller}
We are now ready to prove the equivariant Gutzwiller formula under the assumption of $G$-non-degenerate orbits. According to Section \ref{sec:SpecDist} the Gutzwiller terms are given by the asymptotic expansion of the spectral distribution $\rho_{E, h}^\chi(f)$ for a Schwartz function $f$ with $0\notin \tu{supp}\hat f$. In order to apply Proposition \ref{prop:DiscreteGNonDeg} and to avoid problems concerning the Ehrenfest time, we will additionally need to assume that $\hat f$ is compactly supported. 

As a first step we show that already under the new conditions \ref{hyp:H2prime} and \ref{hyp:GNonDeg} and without any conditions on the group action, the critical set in (\ref{eq:OscIntegralEquiv}) is smooth.

\begin{prop}\label{prop:GutzwSmoothCrit}
 Let $f$ be a Schwartz function with $\supp\hat f\subset [-T,T]\setminus\{0\}$. If for $E\in \R$ the Hamilton dynamics is $G$-non-stationary, the energy shell $\Sigma_E\subset \R^{2n}$ is compact, and all relative periodic orbits in $\Omega_0\cap\Sigma_E$ having a relative period contained in $\supp \hat f$ are $G$-non-degenerate, then the set $\mathcal C_{\phi_E} \cap \left( \R^{2n}\times \supp \hat f\times G\right)$ is a disjoint finite union of smooth submanifolds of dimension dim\,$F$.  
\end{prop}
To prove this proposition, recall from Proposition \ref{prop:CritSetGutzwiller} that the critical set is given by
\[
 \mathcal C_{\phi_E}=\{(z,t,g)\in \R^{2n}\times \R \times G,~~z\in\Omega_0\cap\Sigma_E,~g\Phi_t(z)=z\}
\]
which can be written in terms of the $F$-action as
\[
 \mathcal C_{\phi_E}=\{(z,f)\in  \R^{2n}\times F,~~z\in\Omega_0\cap\Sigma_E,~f\in F_z\}
\]
with $F_z$ being the stabilizer group of $z$ for the $F$-action. We are especially interested in the set of all $t$-values in this stabilizer group. In particular, we need the following
\begin{lem}
 $\mathcal L_z^\tu{rel}:=\{t\in\R,~\exists g\in G\tu{ such that } (t,g) \in F_z\}\subset \R$ is a closed subgroup.
\end{lem}
\begin{proof}
 The subgroup property is clear from the definition, so that it remains to prove the closeness. Suppose that there is a sequence $t_n\in \mathcal L_z^\tu{rel}$ such that $t_n\to t_\infty$. For each $t_n$ there is a $g_n\in G$ such that $(t_n,g_n)\in F_z$. After restricting to a subsequence one can assume that $g_n$ also converges since $G$ is compact. Consequently $(t_n,g_n)$ converges in $F$ and from the closeness of $F_z$ one concludes that $t_\infty\in \mathcal L_z^\tu{rel}$.
\end{proof}
\begin{cor}\label{cor:L_red}
 If $z$ is a $G$-non-stationary, relative periodic point, then there is a $T_z>0$ such that $\mathcal L^\tu{rel}_z =T_z\cdot\Z$.
\end{cor}
\begin{proof}
 All closed subgroups of $\R$ are either empty, $\R$ or of the type $k\Z$ for some $k>0$. As $z$ is relative periodic, we can exclude the empty set. As $z$ is $G$-non-stationary, the trajectory of the Hamilton flow cannot be contained in the $G$-orbit of $z$, so we can exclude $\R$.
\end{proof}
If $z$ is an arbitrary relative periodic point, the corresponding  relative periodic orbit $\gamma_0$ is exactly its orbit under the $F$-action and there is an injective immersion $i:F/F_z\to \gamma_0$. Note that in general an $F$-orbit does not have to be an embedded submanifold in $\R^{2n}$. However, if $z$ is relative periodic, $\mathcal L_z^\tu{rel}=T_z\cdot\Z$ or $\mathcal L_z^\tu{rel}=\R$, and consequently $F/F_z$ is compact. The immersion $i$ is thus additionally closed, and $\gamma_0$ is diffeomorphic to $F/F_z$ as an embedded submanifold. 

Next, note that the critical set $\mathcal C_{\Phi_E}$ is essentially a union of the isotropy bundles of certain relative periodic orbits $\gamma$, which are defined as
\[
\tu{Iso}(\gamma):=\{(z,f)\in  \gamma\times F,~ f\in F_z\}\subset \R^{2n}\times F.
\]
For those isotropy bundles we have the following
\begin{prop}\label{prop:DecompIsoGamma}
  Let $z\in \R^{2n}$ be a $G$-non-stationary, relative periodic point and $\gamma$ its relative orbit. Then $\tu{Iso}(\gamma)$ can be written as
  \begin{equation}\label{eq:decomp_Iso}
  \tu{Iso}(\gamma)=\bigcup\limits_{k\in \Z} M_{k,\gamma},
  \end{equation}
  where
  \begin{equation}\label{eq:M_k_gamma}
    M_{k,\gamma}:=\{(z,kT_\gamma,g)\in  \gamma\times \R\times G:~ g\Phi_{kT_\gamma}(z)=z\}\subset \R^{2n}\times F
  \end{equation}
  are smooth disjoint submanifolds.
\end{prop}
\begin{proof}
 The disjoint decomposition (\ref{eq:decomp_Iso}) follows directly from Corollary \ref{cor:L_red}. It only remains to prove the submanifold property of $M_{k,\gamma}$.  
 
 We will again construct a closed injective immersion, and first recall that for general homogeneous spaces, the isotropy bundle can be written as a quotient (see e.g. \cite[Section 1.11 and 2.4]{DK00}). Consider therefore $F_z$ acting from the right on $F\times F_z$ by right multiplication on the first factor $F$ and by conjugation on the second factor $F_z$. This action is proper and free as the action in the first component is the invertible right multiplication in $F$. Consequently, the quotient $(F\times F_z)/F_z$ carries the structure of a smooth manifold. Furthermore, there is an injective immersion on $\tu{Iso}(\gamma_0)$ which is explicitly given by
 \[
  \mathcal I:\Abb{(F\times F_z)/F_z}{\tu{Iso}(\gamma_0)\subset \R^{2n}\times F}{(f,f_z)F_z}{(fz,ff_zf^{-1}).}
 \]
 For $k\in \Z$ we define the compact sets $F_z^{(k)}:=\{(kT_\gamma, g)\in \R\times G: g\Phi_{k T_\gamma}(z)=z\}\subset F_z$. As $F=\R\times G$ is commutative in its $\R$ component, conjugation with $F_z$ leaves $F_z^{(k)}$ invariant, and $(F\times F^{(k)}_z)/F_z$ is a compact smooth manifold. It is now straightforward to check that $\mathcal I((F\times F^{(k)}_z)/F_z)=M_{k,\gamma}$. So the sets $M_{k,\gamma}$ are images of a compact manifold under an injective immersion and consequently embedded submanifolds.
\end{proof}

We are now able to prove Proposition \ref{prop:GutzwSmoothCrit}.
\begin{proof}[Proof of Proposition \ref{prop:GutzwSmoothCrit}]
 Proposition \ref{prop:DiscreteGNonDeg} implies that the relative periodic orbits are discrete. Together with the assumption that $\Sigma_E$ is compact, this implies that there are only finitely many relative periodic orbits in $(\Gamma_E^{\tu{rel}})_T$, so that 
 \[
  \mathcal C_{\phi_E} \cap \left(\R^{2n}\times\supp \hat f\times G\right) =\bigcup\limits_{\gamma\in (\Gamma_E^{\tu{rel}})_T}\left(\bigcup\limits_{k\in \Z: kT_{\gamma}\in \supp\hat f}M_k,\gamma\right).
 \]
 Proposition \ref{prop:DecompIsoGamma} assures that the $M_{k,\gamma}$ are disjoint, smooth, submanifolds of dimension  dim\,$F$, and as $\tu{supp}\hat f$ is compact, for each $\gamma\in (\Gamma_E^{\tu{rel}})_T $ there are also only finitely many $k\in \Z$ with $k T_\gamma \in\tu{supp}\hat f$.
\end{proof}
In oder to apply the generalized stationary phase theorem to (\ref{eq:OscIntegralEquiv}) it only remains to show the following proposition.
\begin{prop}\label{prop:NonDegHess}
 Let $(z,t,g)\in \mathcal C_{\phi_E}\cap\left( \R^{2n}\times\supp \hat f\times G\right)$, and assume that the dynamics is $G$-non-stationary and that all relative periodic orbits having relative period contained in $\tu{supp}\hat f$ are $G$-non-degenerate. Then $\left(\tu{Hess }\phi_E\right)_{|N_{(z,t,g)}\mathcal C_E}$ is a non-degenerate bilinear form.
\end{prop}
\begin{proof}
Recall that on a $d$-dimensional Riemannian manifold $(M,g)$ the Hessian of a complex valued function $\phi$ on its critical set is a complex valued symmetric 2-form given by
\[
 \tu{Hess }\phi:=\nabla d\phi.
\]
The Hessian is said be non-degenerate on a subspace $V\subset T_xM$ if for $v\in V$ and $\tu{Hess }\phi(v,w)=0$ for all $w\in V$ we have $v=0$. By the metric $g$ we can identify $\tu{Hess }\phi$ with a complex valued $d\times d$ matrix $B=\Re(B)+i\Im(B)$. We denote the real kernel of the Hessian by  $\ker_\R\tu{Hess }\phi(x):=\ker\Re(B)\cap\ker \Im(B)$, and recall that by \cite[Lemma 4.3.5]{Cas05} we have for $(z,t,g)\in \mathcal C_{\phi_E}$ the equivalence 
 \[
  \left(\tu{Hess }\phi_E\right)_{|N_{(z,t,g)}\mathcal C_E}\tu{ is non-degenerate } \Leftrightarrow \tu{ker}_\R ( \tu{Hess }\phi_E(z,t,g))= T_{(z,t,g)}\mathcal C_{\phi_E} .
 \]
By the definition of the critical set, $d\phi_E$ vanishes on $\mathcal C_{\phi_E}$. Consequently for $v\in T_{(z,t,g)}\mathcal C_{\phi_E}$ the one form $\nabla_{v}d\phi_E$ equals zero, which implies that $\tu{ker}_\R ( \tu{Hess }\phi_E)\supset T_{(z,t,g)}\mathcal C_{\phi_E}$. Thus, it suffices to show that in each point $(z,t,g)$ the dimensions of these two linear subspaces of $T_{(z,t,g)}(\R^{2n}\times F)$ coincide.
Furthermore, we can directly use the calculations in \cite[Prop 4.3]{Cas06} by which
\begin{equation}\label{eq:realKernel}
 \begin{array}{rcl}
 \tu{ker}_\R ( \tu{Hess }\phi_E)&=&\{(\alpha,\tau,A)\in \R^{2n}\times\R\times \g:~\alpha\perp\nabla H(z),\alpha\perp J\g z,\\
 &&\tau J \nabla H(z)+Az+\left((g\Phi_t)_{*,z}-Id)\right)\alpha=0\}.
 \end{array}
\end{equation}
Recall that under the $G$-non-stationary hypothesis we have from (\ref{eq:DecompTangentSpaceEquiv}) the decomposition 
\[
 T_z\R^{2n}=J\f z\operp \f z\operp\mathcal R,
\]
where $J\f z$ and $\f z$ both have dimensions equal to $\dim F- \dim F_z$. According to this decomposition we write $\alpha = (\alpha_1,\alpha_2,\alpha_R)$. The first two conditions $\alpha\perp\nabla H(z),\alpha\perp J\g z$ in (\ref{eq:realKernel}) immediately imply $\alpha_1=0$. As $\tau J \nabla H(z)+Az\in \f z$, the third condition in (\ref{eq:realKernel}) implies, after using the general form (\ref{eq:GeneralFormDgPhi}) of $(g\Phi_t)_{*,z}$, that
\[
 (\mathcal P-\mathds 1)\alpha_R=0,
\]
and from the $G$-non-degenerate orbit property we directly conclude that $\alpha_R=0$. The third condition therefore reduces to 
\[
 \tau J \nabla H(z)+Az +(\mathcal B-1)\alpha_2 = 0.
\]
It forms a system of $\dim F-\dim F_z$ linear equations with $2\dim F-\tu{dim} F_z$ variables. From the $G$-non-stationary condition it follows that this system of equations has full rank. Consequently we obtain $\dim \tu{ker}_\R ( \tu{Hess }\phi_E)=\dim F=\dim \mathcal C_{\phi_E}$ which finishes the proof of Proposition \ref{prop:NonDegHess}.
\end{proof}

Taking everything together we finally arrive at
\begin{thm}\label{thm:GutzwillerAsym}
 Let $H$ be a Hamiltonian fulfilling Hypothesis \ref{hyp:regularityH} and \ref{hyp:closed} which is invariant under the compact symmetry group $G\subset O(n)$. Let furthermore $f\in \mathcal S(\R)$ be such that $\supp \hat f \subset [-T,T]\setminus\{0\}$ for some $T>0$. If the Hamilton dynamics is $G$-non-stationary (Hypothesis \ref{hyp:H2prime}) for a given energy $E\in [E_1,E_2]$  and all relative periodic orbits in $\Omega_0\cap\Sigma_E$ having a relative period contained in $\supp \hat f$ are $G$-non-degenerate (Hypothesis \ref{hyp:GNonDeg}), then $\rho_{E, h}^\chi(f)$ has a complete asymptotic expansion in powers of $ h$ given by
 \begin{eqnarray*}\nonumber
 \rho_{E, h}^{\chi}(f) =\frac{d_\chi}{2\pi}\sum\limits_{\gamma\in (\Gamma_E^{\tu{rel}})_T}&\Bigg[&\sum\limits_{k\in \Z} e^{\phi_{k,\gamma}}\hat f(kT_\gamma)\Bigg(\int\limits_{M_{k,\gamma}} \overline{\chi(g)}d(z,t,g)d\sigma_{M_{k,\gamma}}(z,t,g)\Bigg)\\
 &&+\sum\limits_j  h^ja_{j,\gamma,k}\Bigg].
\end{eqnarray*}
Here $M_{k,\gamma}$ was defined in (\ref{eq:M_k_gamma}), $d\sigma_{M_{k,\gamma}}$ is the induced measure on this smooth submanifold, $\phi_{k,\gamma}={\int\limits_0^{kT_\gamma}p_s\dot q_s ds}$ is the constant value the phase function takes on $M_{k,\gamma}$, and the density which has to be integrated is given by
\[
d(z,t,g)=\tu{det}_*^{-1/2}\left(\frac{\tu{Hess }(\phi_E)_{|N\mathcal C_{\phi_E}}}{i}\right)\tu{det}_*^{-1/2}\left(\frac{A+iB-i(C+iD)}{i}\right)
\]
where the matrices $A,B,C,D$ are as in Proposition \ref{prop:OscIntGutzwiller}.

All the lower order coefficients $a_{j,\gamma,k}$ are tempered distributions with support in $kT_\gamma$ applied to $\hat f$,  which can in principle be calculated from the stationary phase approximation.
\end{thm}
\begin{rem}
 Again, the case of a compact symmetry group $G\subset GL(n,\R)$ can be treated by conjugating to an orthogonal subgroup (see the discussion in Section \ref{sec:group_action}).
\end{rem}
\begin{proof}
 Proposition \ref{prop:OscIntGutzwiller} allows us to write the spectral distribution as an oscillating integral. Proposition \ref{prop:GutzwSmoothCrit} together with Proposition \ref{prop:NonDegHess} assure that the critical set is smooth and that the transversal Hessian is non-degenerate. We can therefore apply the generalized stationary phase theorem (see e.g. \cite[Theorem 3.3]{CRR99}) to (\ref{eq:OscIntegralEquiv}) finishing the proof of Theorem \ref{thm:GutzwillerAsym}.
\end{proof}

\end{document}